\documentclass[conference]{IEEEtran}
\usepackage{cite}
\usepackage{amsmath,amssymb,amsfonts,mathtools}
\usepackage{graphicx}
\usepackage{textcomp}
\usepackage{xcolor}
\usepackage{url}
\usepackage{microtype}
\usepackage{hyperref}

\newtheorem{theorem}{Theorem}
\newtheorem{proposition}{Proposition}
\newtheorem{corollary}{Corollary}
\newtheorem{assumption}{Assumption}
\newtheorem{remark}{Remark}

\newcommand{\R}{\mathbb{R}}

\newcommand{\E}{\mathbb{E}}
\newcommand{\Prob}{\mathbb{P}}
\newcommand{\1}{\mathbf{1}}
\newcommand{\abs}[1]{\lvert #1\rvert}
\newcommand{\Leb}{\mathrm{Leb}}

\begin{document}

\title{Arbitrage and the Stability of AMM Price Tracking}

\author{Peihao Li, Nadia Dahmani, and Wenqi Cai%
\thanks{This work was not supported by any organization}%
\thanks{P. Li is with the Computer, Electrical and Mathematical Science and Engineering Division (CEMSE), King Abdullah University of Science and Technology, Kingdom of Saudi Arabia {\tt\small peihao.li@kaust.edu.sa}}%
\thanks{N. Dahmani is with the College of Technological Innovation, Zayed University, United Arab Emirates {\tt\small Nadia.Dahmani@zu.ac.ae}}%
\thanks{W. Cai is with New York University Abu Dhabi {\tt\small wenqi.cai@nyu.edu}}%
}

\maketitle

\begin{abstract}
Automated market makers (AMMs) quote prices from pool state rather than from a limit order book. AMM pools often stay close to a reference price because arbitrageurs correct profitable mispricing. A large part of decentralized finance therefore relies on a simple economic premise: once the AMM price drifts away from the reference price, arbitrage incentives push it back. This paper studies when that premise is strong enough to guarantee block-scale stability. We model the gap between the reference price and the AMM price as a stochastic tracking error, treat arbitrage as the corrective input, and place blockchain execution inside the loop through fees, discrete blocks, transaction ordering, delays, and transaction failure. The detailed execution layer is reduced to the total successful correction confirmed in each block. Under a block-level correction condition, we prove geometric ergodicity of the tracking error and obtain explicit one-step bounds that connect tracking quality to liquidity and execution quality. We also show in a constant-product example how fees, fixed execution costs, and local liquidity map into the no-trade band and the optimal corrective trade. Finally, we build empirical proxies for the theorem quantities from realized block data and use them to organize reduced and mechanism-focused simulations whose comparative statics are consistent with the theory. The contribution is to turn a basic economic intuition behind decentralized finance into a quantitative stability statement together with a tractable calibration interface.
\end{abstract}

\section{Introduction}
\label{sec:intro}

Automated market makers quote prices from pool state rather than from a limit order book, so movements in the reference price create deviations in the AMM quote that arbitrage trades subsequently correct. Much of decentralized finance therefore rests on a simple premise: if the AMM price moves away from the reference price, arbitrage incentives pull it back. This premise is not a side issue. It underlies the practical use of AMMs as trading infrastructure, as price signals, and as the benchmark behind liquidity provider loss calculations \cite{angeris2021analysis,angeris2020oracles,milionis2022lvr}. On a blockchain, however, correction occurs at block scale rather than continuously, and both its timing and magnitude are affected by liquidity, costs, transaction ordering, delay, and transaction failure. In concentrated-liquidity pools, active liquidity varies across price regions, so the strength of correction is state dependent \cite{uniswapv3whitepaper,hashemseresht2022cl,heimbach2022distress,fritsch2024measuring,zhu2024liquidity}.

Existing work does not give the result needed here. Early AMM theory shows that constant function market makers can closely track a reference price under common conditions and can support oracle correctness \cite{angeris2021analysis,angeris2020oracles}. Other work studies arbitrage profits, fee effects, liquidity provider losses, and stylized AMM price dynamics \cite{milionis2022lvr,milionis2023fees,najnudel2024arbdriven,canidio2023batching,angeris2023geometry,zang2026dynamic}. A separate literature shows that MEV, frontrunning, latency races, non-atomic arbitrage, execution priority, and transaction failure determine how much correction is actually realized on chain \cite{daian2020flashboys,torres2021frontrunner,zhou2021hft,adams2024mevslip,heimbach2024nonatomic,fritsch2024timeadv,gogol2024l2arb,gogol2025priorityfails,hautsch2024trust,capponi2026price,adadurov2026subslots}. What is missing, to the best of our knowledge, is a block-scale stochastic closed-loop model that joins these views and proves that the tracking error remains stable once blockchain execution is part of the loop.

This paper fills that gap. We study AMM price tracking at block scale, with arbitrage as the corrective input and blockchain execution inside the loop. The execution layer is summarized by the total successful correction delivered in each block. That summary keeps the features that matter for tracking, including profitability thresholds, local liquidity, ordering, delays, and transaction failure, while avoiding unnecessary detail in the state. Within this framework, we prove sufficient conditions for stochastic stability of the tracking error and derive explicit one-step bounds that connect the rate of correction to liquidity and execution quality. We also show in a constant-product example how fees, fixed execution costs, and reserve depth map into the economic no-trade band and the optimal corrective trade. Finally, we build empirical proxies for the theorem quantities from realized block data and use them to organize reduced and mechanism-focused simulations. The main contribution is therefore a stability theorem for block-scale AMM price tracking together with a tractable way to interpret its parameters in data and in simulation, without claiming a sharp empirical validation of the theorem itself.

\paragraph*{Code availability.}
Code and experiment scripts are available at \url{https://github.com/P-HOW/AMM_price_stability}.

\section{Closed-Loop Model at Block Scale}
\label{sec:model}

We work with log-prices at the block scale. Let $P_n$ denote the reference log-price observed immediately before block $n$ is executed, and let $\tilde P_n$ denote the AMM log-price after the transactions in block $n$ have been processed. The reference signal evolves according to
\begin{equation}
    P_{n+1}=P_n+w_{n+1},
    \label{eq:refprice}
\end{equation}
where $\{w_n\}$ is an exogenous disturbance sequence.

Define the pre-execution and post-execution tracking errors by
\begin{equation}
    x_n \coloneqq P_n-\tilde P_{n-1},
    \qquad
    z_n \coloneqq P_n-\tilde P_n.
    \label{eq:mispricingdefs}
\end{equation}
Then
\begin{equation}
    x_{n+1}=z_n+w_{n+1}.
    \label{eq:stateupdate}
\end{equation}
Thus $x_n$ is the gap seen by arbitrageurs at the start of block $n$, while $z_n$ is the gap left after the arbitrage transactions included in that block have either executed or reverted.

Let $\mathbf{s}_n\in\mathcal S$ denote the pool state and let $g_n\in\mathcal G$ denote the realized block execution environment. We write $\Theta_n=(g_n,\ell_n)$ for the pair formed by execution conditions and local liquidity. To quantify how strongly the pool resists corrective trading, let $P_{\mathrm{exec},s}(q;\mathbf{s})$ denote the effective execution price of a small trade of size $q$ in corrective direction $s\in\{-1,+1\}$ when the pool state is $\mathbf{s}$. We let $\ell_n>0$ denote any measurable local liquidity proxy that is monotone with inverse local price impact, for example
\begin{equation}
    \ell_n
    \coloneqq
    \left|
        \frac{\partial P_{\mathrm{exec},s}}{\partial q}(0;\mathbf{s}_n)
    \right|^{-1}
    \label{eq:liqproxy}
\end{equation}
whenever this derivative exists. In constant-product pools $\ell_n$ is proportional to reserve depth. In concentrated-liquidity pools it is proportional to active in-range liquidity. The role of $\ell_n$ is purely local: it captures the correction authority of marginal arbitrage around the current state.

A corrective trade is only worth submitting when the gap is large enough to clear fees, price impact, and execution costs. We write $\gamma(\ell,g)$ for a measurable band parameter that summarizes the corresponding economic no-trade radius at local liquidity level $\ell$ and execution condition $g$. The theorem below only needs $\gamma$ as a measurable band parameter. The next proposition shows that, in a canonical AMM family, this band is explicit and quantitative.

\begin{proposition}[Constant-product example]
\label{prop:cpband}
Consider a two-asset constant-product AMM with reserves $(x,y)$, invariant $xy=k$, fee multiplier $\eta\in(0,1)$, pool price
\[
    P\coloneqq \frac{y}{x},
\]
external reference price $P^\star$, and fixed numeraire execution cost $c_f(g)\ge 0$. Let direction $+1$ denote the corrective trade that inputs asset $0$ into the pool, and direction $-1$ the corrective trade that inputs asset $1$ into the pool. Then the optimal directional inputs are
\begin{align}
    q_{+1}^\star
    &=
    \left(
        \frac{\sqrt{\eta x y/P^\star}-x}{\eta}
    \right)_+,
    \label{eq:qplusstar}
    \\
    q_{-1}^\star
    &=
    \left(
        \frac{\sqrt{\eta x y P^\star}-y}{\eta}
    \right)_+,
    \label{eq:qminusstar}
\end{align}
and the corresponding optimal directional profits are
\begin{align}
    H_{+1}(P^\star;x,y,g)
    &=
    \left(
        \sqrt{y}-\sqrt{\frac{xP^\star}{\eta}}
    \right)_+^2
    - c_f(g),
    \label{eq:Hpluscp}
    \\
    H_{-1}(P^\star;x,y,g)
    &=
    \left(
        \sqrt{xP^\star}-\sqrt{\frac{y}{\eta}}
    \right)_+^2
    - c_f(g),
    \label{eq:Hminuscp}
\end{align}
In particular, when $c_f(g)=0$, profitable arbitrage exists if and only if
\[
    P^\star \notin [\eta P, P/\eta].
\]
Equivalently, the fee-only no-trade radius in log-price units is
\begin{equation}
    \gamma_0 = -\log \eta.
    \label{eq:feeonlyband}
\end{equation}
With $c_f(g)>0$, the no-trade radius is the unique root of $H_{+1}=0$ or $H_{-1}=0$ in the active direction, hence is still finite and quantitative.
\end{proposition}

\begin{proof}
For direction $+1$, the output of a token-$0$ input $q$ is
\[
    \frac{y\eta q}{x+\eta q},
\]
so the numeraire profit is
\[
    \Pi_{+1}(q)=\frac{y\eta q}{x+\eta q}-P^\star q-c_f(g).
\]
Differentiating and optimizing over $q\ge 0$ yields \eqref{eq:qplusstar}. Substituting the optimizer back gives \eqref{eq:Hpluscp}. The computation for direction $-1$ is analogous, using the output
\[
    \frac{x\eta q}{y+\eta q}
\]
valued at the reference price $P^\star$. The fee-only threshold follows by setting $c_f(g)=0$ and checking when the positive part in \eqref{eq:Hpluscp} or \eqref{eq:Hminuscp} is strictly positive.
\end{proof}

\begin{remark}
Proposition~\ref{prop:cpband} is included for two reasons. First, it shows that the economic no-trade radius is explicit in a canonical AMM family. Second, it provides a direct way to compute the predicted optimal arbitrage size and trigger condition block by block from reserve data and execution costs. Those quantities will also be used later in a custom mechanism-focused simulation.
\end{remark}

For a submitted transaction, the executable residual band is the larger of its economic threshold and its explicit slippage guard. If transaction $r$ was prepared under submission state $(\ell_{n,r}^{\mathrm{sub}},g_{n,r}^{\mathrm{sub}})$ and carries a user-specified slippage tolerance $\zeta_{n,r}\ge 0$, define
\begin{equation}
    \gamma_{n,r}
    \coloneqq
    \max\left\{
        \gamma(\ell_{n,r}^{\mathrm{sub}},g_{n,r}^{\mathrm{sub}}),
        \zeta_{n,r}
    \right\}.
    \label{eq:gammar}
\end{equation}
Thus a transaction can execute only if enough residual discrepancy remains when it reaches its realized block position to keep it both profitable and within slippage protection. This is exactly the quantity that appears later in the sequential execution rule.

Let $\mathcal M$ denote the collection of finite ordered lists of triples $(s,u,\gamma)$ with $s\in\{-1,+1\}$, $u>0$, and $\gamma\ge 0$. Once block $n$ is realized, we summarize the included arbitrage transactions by an ordered list
\[
    \mathcal M_n=
    \big((s_{n,1},u_{n,1},\gamma_{n,1}),\dots,(s_{n,K_n},u_{n,K_n},\gamma_{n,K_n})\big)\in\mathcal M.
\]
Here $K_n$ is the number of included arbitrage transactions, $s_{n,r}$ is the corrective direction, $u_{n,r}$ is the realized amount by which the transaction would reduce the current error magnitude if it succeeds, and $\gamma_{n,r}$ is the transaction's executable residual band.

\paragraph*{Sequential execution rule.}
Fix a realized pre-execution error $x_n$. Define recursively
\begin{equation}
    R_{n,0}\coloneqq x_n,
    \label{eq:R0}
\end{equation}
and for $r=1,\dots,K_n$,
\begin{align}
    \chi_{n,r}
    &\coloneqq
    \1\!\left\{
        s_{n,r}R_{n,r-1}>0,\
        \abs{R_{n,r-1}}\ge \gamma_{n,r}+u_{n,r}
    \right\},
    \label{eq:successindicator}
    \\
    R_{n,r}
    &\coloneqq
    R_{n,r-1}-\chi_{n,r}s_{n,r}u_{n,r}.
    \label{eq:sequentialupdate}
\end{align}
The indicator $\chi_{n,r}$ equals one exactly when the included transaction points in the same direction as the remaining error and the remaining error is large enough to satisfy the transaction's executable residual band. Otherwise the transaction reverts and the tracking error does not change.

Given an initial error $x\in\R$ and an ordered list $M\in\mathcal M$, we write $\Psi(x,M)$ for the final error after applying the same recursion and $C(x,M)$ for the total successful correction delivered by that list. In particular,
\begin{equation}
    z_n\coloneqq \Psi(x_n,\mathcal M_n),
    \qquad
    C_n\coloneqq C(x_n,\mathcal M_n)=\sum_{r=1}^{K_n}\chi_{n,r}u_{n,r}.
    \label{eq:postoperator}
\end{equation}
The scalar $C_n$ is the realized correction that matters for the block-scale tracking loop.

The next proposition is the key reduction step. It shows that the full ordered execution layer enters the tracking dynamics only through the realized correction $C_n$ and a uniform dead-zone cap.

\begin{proposition}
\label{prop:servicebound}
Suppose there exists a finite constant $\bar\gamma$ such that
\begin{equation}
    \gamma_{n,r}\le \bar\gamma
    \label{eq:bandcap}
\end{equation}
almost surely for every included transaction. Then for every block $n$,
\begin{equation}
    \bigl(\abs{z_n}-\bar\gamma\bigr)_+
    \le
    \Big(\bigl(\abs{x_n}-\bar\gamma\bigr)_+ - C_n\Big)_+.
    \label{eq:excessservice}
\end{equation}
In particular,
\begin{equation}
    \abs{z_n}\le \abs{x_n}.
    \label{eq:monotonebound}
\end{equation}
\end{proposition}

\begin{proof}
Write $e_0\coloneqq (\abs{x_n}-\bar\gamma)_+$. We claim that for each $r=0,1,\dots,K_n$,
\begin{equation}
    \abs{R_{n,r}}
    \le
    \bar\gamma+
    \Big(e_0-\sum_{j=1}^{r}\chi_{n,j}u_{n,j}\Big)_+.
    \label{eq:inductionclaim}
\end{equation}
For $r=0$ this is immediate from $R_{n,0}=x_n$. Assume it holds at step $r-1$. If $\chi_{n,r}=0$, then $R_{n,r}=R_{n,r-1}$ and the bound is unchanged. If $\chi_{n,r}=1$, then $s_{n,r}R_{n,r-1}>0$ and $\abs{R_{n,r-1}}\ge \gamma_{n,r}+u_{n,r}\ge u_{n,r}$, so the transaction removes exactly $u_{n,r}$ units of error magnitude:
\[
    \abs{R_{n,r}}=\abs{R_{n,r-1}}-u_{n,r}.
\]
Using the induction hypothesis at step $r-1$ and the inequality $(a)_+-u\le (a-u)_+$ for $u\ge 0$ gives
\[
    \abs{R_{n,r}}
    \le
    \bar\gamma+
    \Big(e_0-\sum_{j=1}^{r}\chi_{n,j}u_{n,j}\Big)_+.
\]
This proves \eqref{eq:inductionclaim}. Setting $r=K_n$ and using \eqref{eq:postoperator} yields \eqref{eq:excessservice}. The monotonicity bound \eqref{eq:monotonebound} follows directly from \eqref{eq:sequentialupdate}: each successful trade reduces $\abs{R_{n,r}}$, and each failed trade leaves it unchanged.
\end{proof}

The proposition isolates the two quantities that remain in the theorem: the dead-zone cap $\bar\gamma$ and the realized block correction $C_n$. Everything else in the execution layer matters only through how it shapes those two objects.

\section{Main Stability Result}
\label{sec:stability}

\subsection{Assumptions}

The theorem only needs a block-level closure of the dynamics, not a full model of the searcher game.

\begin{assumption}[Block-level closure and dead-zone cap]
\label{ass:model}
The process $\{x_n\}$ is generated as follows. Given the current error $x_n=x$, the realized block environment $\Theta_n=(g_n,\ell_n)\in\mathcal G\times(0,\infty)$ is drawn from a time-homogeneous stochastic kernel $\mathsf K(\cdot\mid x)$. Given $(x_n,\Theta_n)=(x,g,\ell)$, the ordered arbitrage list $\mathcal M_n$ is drawn from a time-homogeneous kernel $\mathsf Q(\cdot\mid x,g,\ell)$. The post-execution error is $z_n=\Psi(x_n,\mathcal M_n)$ and the next state satisfies \eqref{eq:stateupdate}. In addition, there exists a finite constant $\bar\gamma$ such that \eqref{eq:bandcap} holds almost surely, so Proposition~\ref{prop:servicebound} applies in every realized block.
\end{assumption}

\begin{assumption}[Disturbance]
\label{ass:noise}
The innovations $\{w_n\}$ are independent and identically distributed, admit a continuous density $f$ on $\R$ that is strictly positive everywhere, and satisfy
\[
    \E\big[e^{\eta \abs{w_1}}\big] < \infty
\]
for some $\eta>0$.
\end{assumption}

\begin{assumption}[Large-error correction condition]
\label{ass:service}
There exist constants $x_{\star}>\bar\gamma$, $\lambda_{\star}\in(0,1]$, and $p_{\star}\in(0,1]$ such that whenever $\abs{x_n}\ge x_{\star}$,
\begin{equation}
    \Prob\Big(C_n \ge \lambda_{\star}\big(\abs{x_n}-\bar\gamma\big)\ \Big|\ x_n\Big)
    \ge p_{\star}.
    \label{eq:serviceassumption}
\end{equation}
\end{assumption}

Assumption~\ref{ass:service} is the economic heart of the paper. It says that once the gap is sufficiently large, arbitrage does not need to clear all of it in every block. It is enough that, with probability at least $p_{\star}$, the realized block correction removes at least a fraction $\lambda_{\star}$ of the gap outside the dead zone.

\subsection{Markov Properties}

\begin{proposition}
\label{prop:smallset}
Under Assumptions~\ref{ass:model} and \ref{ass:noise}, the tracking error process $\{x_n\}$ is a time-homogeneous, Lebesgue-irreducible, strongly aperiodic Markov chain. Moreover, every compact interval is one-step small.
\end{proposition}

\begin{proof}
Under Assumption~\ref{ass:model}, once $x_n=x$ is fixed, the law of $(\Theta_n,\mathcal M_n)$ depends on the past only through $x$, and Assumption~\ref{ass:noise} makes $w_{n+1}$ independent of the previous history. Hence the conditional law of $x_{n+1}=\Psi(x_n,\mathcal M_n)+w_{n+1}$ depends on the past only through $x_n$, so $\{x_n\}$ is Markov and time homogeneous.

Fix $x\in\R$ and a measurable set $A\subseteq\R$ with $\Leb(A)>0$. Since $f$ is strictly positive everywhere and $\Psi(x,\mathcal M_n)$ is finite almost surely,
\[
    \int_A f\bigl(y-\Psi(x,\mathcal M_n)\bigr)\,dy>0
\]
almost surely. Taking conditional expectation given $x_n=x$ yields $\Prob(x_{n+1}\in A\mid x_n=x)>0$, so the chain is Lebesgue irreducible.

Now fix a compact interval $C=[-R,R]$ and let $I=[-1,1]$. If $x_n=x\in C$, Proposition~\ref{prop:servicebound} gives $\abs{z_n}\le \abs{x}\le R$ almost surely. Therefore, for any measurable set $A\subseteq\R$,
\[
    \Prob(x_{n+1}\in A\mid x_n=x)
    =
    \E\!\left[
        \int_A f(y-z_n)\,dy
        \,\middle|\,
        x_n=x
    \right].
\]
Define $m_R\coloneqq \inf\{f(u):\abs{u}\le R+1\}$. By continuity and strict positivity of $f$, one has $m_R>0$. If $y\in I$ and $\abs{z_n}\le R$, then $\abs{y-z_n}\le R+1$, hence $f(y-z_n)\ge m_R$. Consequently,
\[
    \Prob(x_{n+1}\in A\mid x_n=x)
    \ge
    m_R\,\Leb(A\cap I)
    \qquad
    \text{for all }x\in C.
\]
Thus every compact interval is one-step small, and the same minorization implies strong aperiodicity.
\end{proof}

\subsection{Stability Theorem}

\begin{theorem}
\label{thm:drift}
Under Assumptions~\ref{ass:model}--\ref{ass:service}, there exist constants $\alpha\in(0,\eta)$, $\rho\in(0,1)$, $B<\infty$, and $R<\infty$ such that the Lyapunov function
\[
    V(x)\coloneqq e^{\alpha \abs{x}}
\]
satisfies
\begin{equation}
\E\big[V(x_{n+1}) \mid x_n=x\big]
\le \rho V(x)+B\,\1\{\abs{x}\le R\},
\label{eq:drift}
\end{equation}
for all $x\in\R$.
Consequently, the chain admits a unique invariant probability measure $\pi_x$, is positive Harris recurrent, and converges to $\pi_x$ geometrically fast in the norm weighted by $V$. In particular,
\[
    \int_{\R} e^{\alpha \abs{x}}\,\pi_x(dx)<\infty.
\]
\end{theorem}

\begin{proof}
The proof follows the economic logic of the model. Outside the dead zone, a successful arbitrage episode contracts the current gap. Assumption~\ref{ass:service} says that such a contraction arrives with uniformly positive probability once the gap is large enough. The disturbance can then be treated as an additive shock on top of that correction mechanism.

For $\alpha\in(0,\eta)$ define the disturbance moment
\[
    M_w(\alpha)\coloneqq \E[e^{\alpha\abs{w_1}}]<\infty.
\]
Fix $x\in\R$ and write
\[
    e(x)\coloneqq \big(\abs{x}-\bar\gamma\big)_+.
\]
Let
\[
    \mathsf A(x)\coloneqq \{C_n\ge \lambda_{\star}e(x)\}
\]
be the event that block $n$ delivers the large correction promised by Assumption~\ref{ass:service}.

On $\mathsf A(x)$, Proposition~\ref{prop:servicebound} gives
\[
    \abs{z_n}
    \le
    \bar\gamma+\bigl(e(x)-\lambda_{\star}e(x)\bigr)_+
    =
    \bar\gamma+(1-\lambda_{\star})e(x).
\]
If $\abs{x}\ge x_{\star}$, then $x_{\star}>\bar\gamma$ implies $e(x)=\abs{x}-\bar\gamma$, and therefore
\begin{equation}
    \abs{z_n}
    \le
    (1-\lambda_{\star})\abs{x}+\lambda_{\star}\bar\gamma.
    \label{eq:serviceeventbound}
\end{equation}
This is the key contraction estimate: on the good event, the post-execution gap is at most an affine contraction of the pre-execution gap.

Using \eqref{eq:stateupdate}, the triangle inequality, and the independence of $w_{n+1}$ from block-$n$ service variables, we obtain on $\mathsf A(x)$,
\begin{equation}
\begin{aligned}
\E\big[V(x_{n+1}) \,\big|\,
   x_n=x,\mathsf A(x)\big]
&= \E\big[e^{\alpha\abs{z_n+w_{n+1}}}
   \,\big|\, x_n=x,\mathsf A(x)\big]
\\
&\le \E\big[e^{\alpha\abs{z_n+w_{n+1}}}
   \,\big|\, x_n=x,\mathsf A(x)\big]
\\
&\le \E\big[e^{\alpha\abs{z_n}+\alpha\abs{w_{n+1}}}
   \,\big|\, x_n=x,\mathsf A(x)\big]
\\
&= \E\big[e^{\alpha\abs{z_n}}
   \,\big|\, x_n=x,\mathsf A(x)\big]\,M_w(\alpha)
\\
&\le M_w(\alpha)\,
   e^{\alpha\big((1-\lambda_{\star})\abs{x}
   +\lambda_{\star}\bar\gamma\big)}.
\end{aligned}
\label{eq:goodbranch}
\end{equation}
On the complement $\mathsf A(x)^c$, we only use the monotonicity part of Proposition~\ref{prop:servicebound}, namely $\abs{z_n}\le\abs{x}$. Hence
\begin{align}
    \E\big[V(x_{n+1})\mid x_n=x,\mathsf A(x)^c\big]
    &\le M_w(\alpha)e^{\alpha\abs{x}}.
    \label{eq:badbranch}
\end{align}

Now mix the two branches. Let
\[
    p(x)\coloneqq \Prob\bigl(\mathsf A(x)\mid x_n=x\bigr).
\]
Combining \eqref{eq:goodbranch} and \eqref{eq:badbranch} yields
\begin{equation}
\begin{aligned}
\E\big[V(x_{n+1})\mid x_n=x\big]
&\le M_w(\alpha)\Big((1-p(x))
   +p(x)e^{\alpha\lambda_{\star}(\bar\gamma-\abs{x})}\Big)
\\
&\qquad\qquad\qquad\qquad\qquad\qquad \cdot V(x).
\end{aligned}
\label{eq:mixture}
\end{equation}
If $\abs{x}\ge x_{\star}$, Assumption~\ref{ass:service} gives $p(x)\ge p_{\star}$, so
\begin{equation}
\begin{aligned}
\E\big[V(x_{n+1})\mid x_n=x\big]
&\le M_w(\alpha)\Big((1-p_{\star})
   +p_{\star}e^{\alpha\lambda_{\star}(\bar\gamma-\abs{x})}\Big)
\\
&\qquad\qquad\qquad\qquad\qquad\qquad \cdot V(x).
\end{aligned}
\label{eq:outsidecompact}
\end{equation}
This inequality already shows the structure of the result. The first term, $M_w(\alpha)(1-p_{\star})$, is the contribution of blocks that do \emph{not} deliver a large correction. The second term decays exponentially in $\abs{x}$ because on the good event the post-execution error is strictly smaller than the pre-execution error.

We now choose the Lyapunov exponent $\alpha$. Since $M_w(0)=1$ and $1-p_{\star}<1$, continuity of $M_w$ at the origin implies that one may pick $\alpha\in(0,\eta)$ small enough that
\[
    a_{\alpha}\coloneqq M_w(\alpha)(1-p_{\star})<1.
\]
Fix any $\rho\in(a_{\alpha},1)$. Because the exponential factor in \eqref{eq:outsidecompact} converges to zero as $\abs{x}\to\infty$, there exists $R\ge x_{\star}$ such that
\[
    M_w(\alpha)
    \Big((1-p_{\star})+p_{\star}e^{\alpha\lambda_{\star}(\bar\gamma-\abs{x})}\Big)
    \le \rho
    \qquad
    \text{for all }\abs{x}>R.
\]
Substituting this into \eqref{eq:outsidecompact} gives the drift bound
\begin{equation}
    \E\big[V(x_{n+1})\mid x_n=x\big]\le \rho V(x)
    \qquad
    \text{for all }\abs{x}>R.
    \label{eq:outsidecompactdrift}
\end{equation}

It remains to control the compact set $\{\abs{x}\le R\}$. If $\abs{x}\le R$, then Proposition~\ref{prop:servicebound} yields $\abs{z_n}\le\abs{x}\le R$. Therefore
\[
    \abs{x_{n+1}}
    =\abs{z_n+w_{n+1}}
    \le \abs{z_n}+\abs{w_{n+1}}
    \le R+\abs{w_{n+1}}.
\]
Hence
\[
    V(x_{n+1})\le e^{\alpha R}e^{\alpha\abs{w_{n+1}}},
\]
and so
\[
    B\coloneqq \sup_{\abs{x}\le R}\E\big[V(x_{n+1})\mid x_n=x\big]
    \le e^{\alpha R}M_w(\alpha)<\infty.
\]
Combining this compact-set bound with \eqref{eq:outsidecompactdrift} proves \eqref{eq:drift}.

Finally, Proposition~\ref{prop:smallset} shows that the chain is Lebesgue irreducible, strongly aperiodic, and has compact small sets. Standard Foster--Lyapunov theory therefore implies positive Harris recurrence and geometric ergodicity \cite{meyn2009markov}. The uniqueness of the invariant measure and the exponential moment bound under $\pi_x$ follow from the same drift criterion.
\end{proof}

The theorem is the main result of the paper. Outside a compact operating tube, the block-scale tracking loop has strict negative exponential drift. The quantity $\lambda_{\star}p_{\star}$ plays the role of an effective correction strength: $\lambda_{\star}$ measures how much of the excess gap is removed when the correction arrives, while $p_{\star}$ measures how reliably that correction is delivered.

\subsection{Two Short Corollaries}

The first corollary is the one-step bound used later in the empirical section.

\begin{corollary}[One-step excess recursion]
\label{cor:excessrecursion}
Let
\begin{equation}
    e_n\coloneqq \big(\abs{x_n}-\bar\gamma\big)_+.
    \label{eq:excessdef}
\end{equation}
Under Assumption~\ref{ass:model},
\begin{equation}
    e_{n+1}
    \le
    (e_n-C_n)_+ + \abs{w_{n+1}}
    \qquad
    \text{almost surely}.
    \label{eq:detexcess}
\end{equation}
If $\mu_w\coloneqq \E[\abs{w_1}]<\infty$, then for any $\lambda\in(0,1]$,
\begin{equation}
    \E[e_{n+1}\mid x_n=x]
    \le
    \bigl(1-\lambda p_{\lambda}(x)\bigr)e(x)+\mu_w,
    \label{eq:excessrecursion}
\end{equation}
where $e(x)=(\abs{x}-\bar\gamma)_+$ and
\[
    p_{\lambda}(x)\coloneqq \Prob\big(C_n\ge \lambda e(x)\mid x_n=x\big).
\]
In particular, under Assumption~\ref{ass:service},
\begin{equation}
    \E[e_{n+1}\mid x_n=x]
    \le
    \bigl(1-\lambda_{\star}p_{\star}\bigr)e(x)+\mu_w,
    \qquad
    \abs{x}\ge x_{\star}.
    \label{eq:excessrecursion_uniform}
\end{equation}
\end{corollary}

\begin{proof}
Equation \eqref{eq:detexcess} follows from Proposition~\ref{prop:servicebound} and the elementary inequality $(\abs{a+b}-\bar\gamma)_+\le (\abs{a}-\bar\gamma)_+ + \abs{b}$. Conditioning on the event $\{C_n\ge \lambda e(x)\}$ then gives \eqref{eq:excessrecursion}. The uniform bound \eqref{eq:excessrecursion_uniform} follows from Assumption~\ref{ass:service}.
\end{proof}

The second corollary is the local certification interface. It is shorter than the main theorem but useful in practice because it turns a local lower bound on realized correction into a local contraction rate.

\begin{corollary}[Local contraction proxy]
\label{cor:localrate}
Fix a measurable bin $\mathcal B\subseteq\R\times\mathcal G\times(0,\infty)$ and a threshold $R>\bar\gamma$. Suppose there exist numbers $p_{\mathcal B}\in(0,1]$ and $\lambda_{\mathcal B}\in(0,1]$ such that
\[
    \mathsf Q\!\Bigl(
        \bigl\{m\in\mathcal M:C(x,m)\ge \lambda_{\mathcal B}(\abs{x}-\bar\gamma)\bigr\}
        \,\Big|\,
        x,g,\ell
    \Bigr)
    \ge p_{\mathcal B}
\]
for all $(x,g,\ell)\in\mathcal B$ with $\abs{x}\ge R$. Then for every $\alpha\in(0,\eta)$,
\begin{equation}
    \E\big[V(x_{n+1})\mid x_n=x,g_n=g,\ell_n=\ell\big]
    \le
    \rho_{\mathcal B}(R,\alpha)V(x),
    \label{eq:explicitrate}
\end{equation}
where
\begin{equation}
    \rho_{\mathcal B}(R,\alpha)
    \coloneqq
    M_w(\alpha)\Bigl((1-p_{\mathcal B}) + p_{\mathcal B}e^{\alpha\lambda_{\mathcal B}(\bar\gamma-R)}\Bigr).
    \label{eq:rhoB}
\end{equation}
Hence $\rho_{\mathcal B}(R,\alpha)<1$ certifies local contraction in exponential moment.
\end{corollary}

\begin{proof}
Repeat the proof of Theorem~\ref{thm:drift} with $p_{\star},\lambda_{\star},x_{\star}$ replaced by $p_{\mathcal B},\lambda_{\mathcal B},R$ and condition on the local event $\{C_n\ge \lambda_{\mathcal B}(\abs{x}-\bar\gamma)\}$.
\end{proof}

\section{Empirical Calibration and Simulation}
\label{sec:emp}

Section~\ref{sec:stability} reduces the closed-loop question to three theorem quantities: the dead-zone cap $\bar\gamma$, the large-error threshold $x_{\star}$, and the large-error correction pair $(\lambda_{\star},p_{\star})$. The empirical section uses realized block data to build simple proxies for those quantities and then feeds them directly into a reduced simulation. The goal is not to re-prove the theorem from data. The goal is to show that the correction pattern required by the theorem is visible on chain and that a simulation driven by the same quantities behaves the way the theory says it should.

\subsection{What We Measure on Chain}

We scan Ethereum mainnet blocks 16800007 through 16803511 and identify 397 simple arbitrage observations with usable pre- and post-execution measurements. Ethereum mainnet is a natural setting for this study because it hosts economically significant AMM trading and arbitrage activity \cite{zhou2021hft,heimbach2024nonatomic,zhu2024liquidity,gogol2024l2arb}. We reconstruct pair-specific event windows around arbitrage activity. For each simple arbitrage observation, we measure a pre-execution gap $e_n^{\mathrm{pre}}$ and a post-execution gap $e_n^{\mathrm{post}}$ against a reference pool, both in log-price units. The realized within-block correction is then
\begin{equation}
    S_n\coloneqq e_n^{\mathrm{pre}}-e_n^{\mathrm{post}}.
    \label{eq:empservice}
\end{equation}
This is the empirical counterpart of the theoretical correction term $C_n$.

We also use the next block to build a disturbance proxy. Writing $e_n=(e_n^{\mathrm{pre}}-\bar\gamma)_+$, the one-step inequality in Corollary~\ref{cor:excessrecursion} suggests comparing the realized next-block excess $e_{n+1}$ with $(e_n-S_n)_+$. The remaining difference is treated as an empirical proxy for the next-block disturbance magnitude.

\subsection{From Data to Theorem Parameters}

The dead-zone proxy is set from the lower end of the observed gap distribution. In the calibration used below, we take
\[
    \bar\gamma \approx 3.82\times 10^{-4},
\]
which is half of the first-quartile median pre-execution gap. This choice is intentionally conservative: it treats the smallest routinely observed gaps as economically inactive for the purpose of the reduced model.

For the large-error threshold, we use the sample median of $e_n^{\mathrm{pre}}$. This split is simple and transparent, and it matches the logic of Assumption~\ref{ass:service}, which only asks for reliable correction once the gap is already meaningfully outside the operating tube.

The remaining quantities come from a direct question: for a given required correction share $\lambda$, how often does a block remove at least that share of the current gap? In symbols, we estimate
\begin{equation}
    \hat p(\lambda)
    \coloneqq
    \Prob\bigl(S_n\ge \lambda e_n^{\mathrm{pre}}\bigr)
    \label{eq:phatlambdatext}
\end{equation}
on the large-error half of the sample. Figure~\ref{fig:sec4-correctioncurve} plots this curve for both large and small errors. The difference is clear. Large gaps are corrected more often, and they are also corrected more deeply.

To keep the baseline rule fixed and transparent, we do not choose $\hat\lambda_{\star}$ after looking at the simulation output. Instead we restrict attention to the reference levels $\{0.25,0.50,0.75\}$ and choose the largest one whose large-error delivery probability still exceeds $70\%$. In our sample this rule selects
\[
    \hat\lambda_{\star}=0.50,
    \qquad
    \hat p_{\star}=0.729,
\]
because the large-error shares at $\lambda=0.25$, $0.50$, and $0.75$ are $85.4\%$, $72.9\%$, and $54.8\%$, respectively. The rule is meant to capture the most ambitious correction target that is still supported by a clear empirical majority of large-error blocks.

\begin{figure}[t]
\centering
\includegraphics[width=\columnwidth]{\detokenize{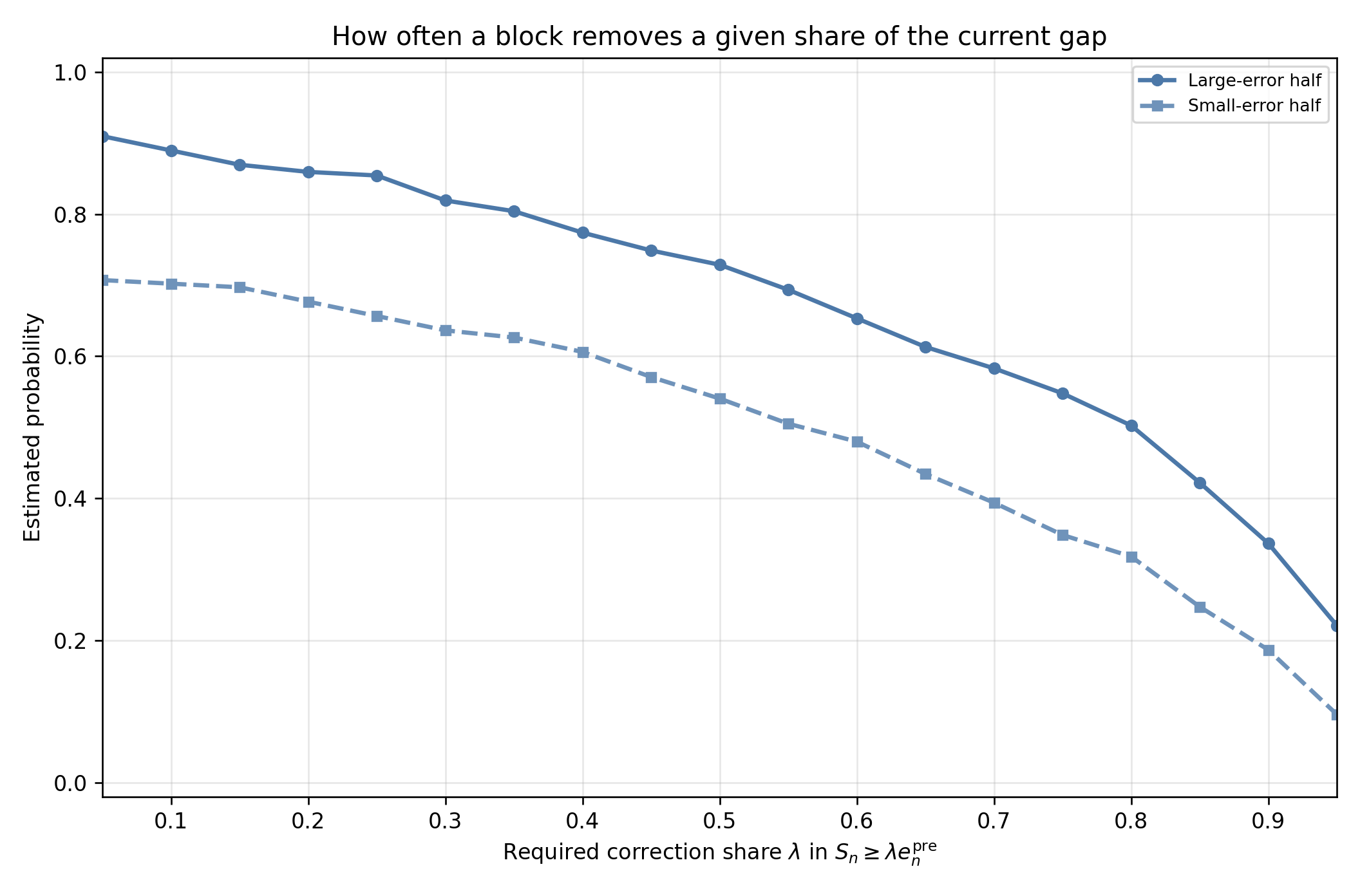}}
\caption{Probability that a block removes at least a given share of the current gap. The vertical axis plots $\hat p(\lambda)=\Prob(S_n\ge \lambda e_n^{\mathrm{pre}})$. Large-error observations dominate small-error observations across the whole range.}
\label{fig:sec4-correctioncurve}
\end{figure}

Table~\ref{tab:sec4-calibration} summarizes the main calibration quantities. Table~\ref{tab:sec4-robust} adds two short robustness checks that matter for interpretation: the quartile ranking is monotone in the way already visible in the data, and the one-step proxy implied by Corollary~\ref{cor:excessrecursion} becomes exact in sample once the empirical disturbance proxy is added.

\begin{table}[t]
\caption{Data-guided theorem quantities used in the simulation.}
\label{tab:sec4-calibration}
\centering
\footnotesize
\setlength{\tabcolsep}{4pt}
\renewcommand{\arraystretch}{0.97}
\begin{tabular}{p{0.63\columnwidth}r}
\hline
Quantity & Value \\
\hline
Dead-zone proxy $\bar\gamma$ & $3.82\times 10^{-4}$ \\
Large-error threshold $x_{\star}$ & median of $e_n^{\mathrm{pre}}$ \\
Large-error correction share $\hat\lambda_{\star}$ & $0.50$ \\
Delivery probability at $\hat\lambda_{\star}$ & $72.9\%$ \\
Large-error share at $\lambda=0.25$ & $85.4\%$ \\
Large-error share at $\lambda=0.75$ & $54.8\%$ \\
Overall positive correction ratio & $81.1\%$ \\
\hline
\end{tabular}
\end{table}

\begin{table}[t]
\caption{Simple robustness checks for the calibration interface.}
\label{tab:sec4-robust}
\centering
\footnotesize
\setlength{\tabcolsep}{4pt}
\renewcommand{\arraystretch}{0.97}
\begin{tabular}{p{0.70\columnwidth}r}
\hline
Statistic & Value \\
\hline
Overall share $e_{n+1}\le (e_n-S_n)_+$ & $86.4\%$ \\
Overall share $e_{n+1}\le (e_n-S_n)_+ + \widehat w_n$ & $100.0\%$ \\
Tracked-token share $e_{n+1}\le (e_n-S_n)_+ + \widehat w_n$ & $100.0\%$ \\
Q1 / Q4 positive-correction ratio & $66.0\%$ / $97.0\%$ \\
Q1 / Q4 median relative correction & $0.473$ / $0.874$ \\
\hline
\end{tabular}
\end{table}

The tracked-token subsample formed by WETH/USDC, WETH/USDT, and USDT/USDC tells the same story, so the large-error correction pattern is not driven by a single pair. The next-block check also lines up with the one-step recursion. The raw proxy
\[
    e_{n+1}\le (e_n-S_n)_+
\]
holds for $86.4\%$ of the overall sample. Once we add the empirical disturbance proxy $\widehat w_n$, the share rises to $100.0\%$. The same happens on the tracked-token subsample, where the corresponding share with the disturbance proxy is also $100.0\%$. This is exactly the pattern suggested by Corollary~\ref{cor:excessrecursion}: realized correction explains most of the one-step movement in the gap, and the rest is well described as a next-block innovation term.

\subsection{Reduced Simulation Driven by the Same Parameters}

The theorem does not need a detailed model of the searcher game. It only needs a dead-zone cap and a large-error correction rule. The simulation therefore follows the same logic. We simulate the block-scale update
\[
    x_{n+1}=z_n+w_{n+1},
\]
keep the sequential execution rule of Section~\ref{sec:model}, and let the realized correction in each block depend on three ingredients: the current gap, a local liquidity state, and an execution-quality state. Above the threshold $x_{\star}$, the baseline correction law is anchored at the empirical pair $(\hat\lambda_{\star},\hat p_{\star})$. Below that threshold, the correction law is interpolated from the smaller-gap part of the sample. Strong-service and weak-service cases are then created by shifting the same quantities up and down around the empirical baseline.

Figure~\ref{fig:sec4-tracking} shows the resulting time paths under a common disturbance sequence with intermittent shocks. The baseline case keeps the AMM price close to the reference price and typically returns to the operating tube quickly after shocks. Stronger correction improves that behavior further, while weaker correction leaves the process outside the tube for longer. Mean excess error is $6.55\times 10^{-3}$ in the weak-service case, $3.13\times 10^{-3}$ in the baseline case, and $1.71\times 10^{-3}$ in the strong-service case. The message is the same as in Theorem~\ref{thm:drift}: when the large-error correction pair improves, tracking becomes tighter and recovery becomes faster.

\begin{figure}[t]
\centering
\includegraphics[width=\columnwidth]{\detokenize{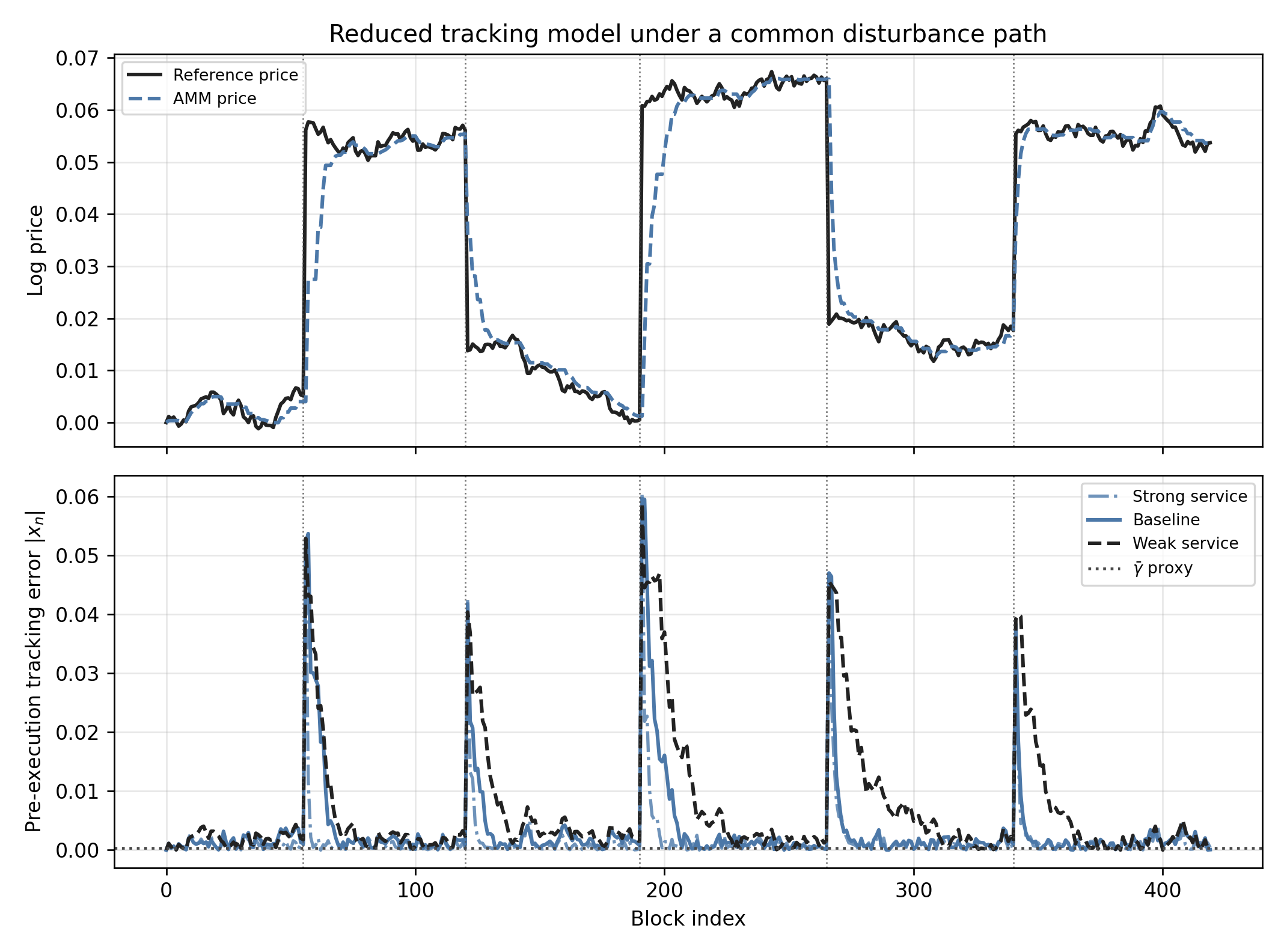}}
\caption{Reduced simulation driven by the data-guided theorem quantities. Top: reference and AMM log-prices in the baseline case under intermittent shocks. Bottom: pre-execution tracking error under strong-service, baseline, and weak-service cases.}
\label{fig:sec4-tracking}
\end{figure}

To make the role of the theorem parameters explicit, we also sweep the pair $(\lambda,p)$ directly. For each point on the grid, we rerun the reduced model and record the mean excess error. Figure~\ref{fig:sec4-grid} shows a clean monotone picture: larger correction shares and more reliable delivery both improve tracking. The empirically calibrated baseline sits well inside the low-error region. The same figure also overlays the simple local contraction proxy from Corollary~\ref{cor:localrate}, which gives a direct visual bridge between the theorem and the simulation.

\begin{figure}[t]
\centering
\includegraphics[width=\columnwidth]{\detokenize{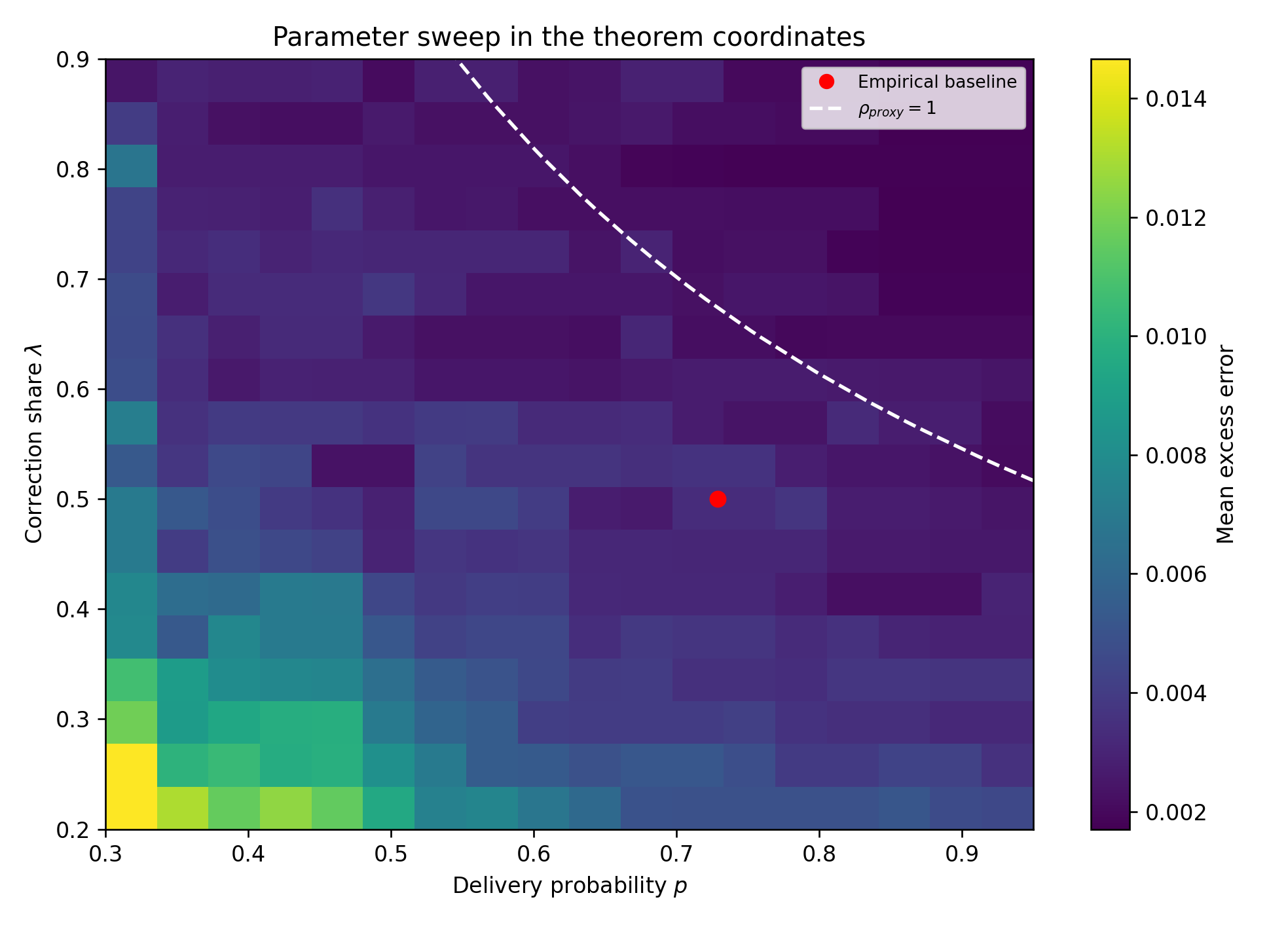}}
\caption{Parameter sweep in the theorem coordinates. Each point corresponds to a large-error correction share $\lambda$ and delivery probability $p$. The color shows the mean excess error in the reduced simulation. The marker indicates the empirical baseline, and the dashed contour shows the approximate boundary where the local contraction proxy equals one.}
\label{fig:sec4-grid}
\end{figure}

\subsection{Constant-Product Mechanism Simulation}

Finally, we add a custom mechanism simulation that does not use chain calibration. This exercise is included to support the mechanism interpretation of Proposition~\ref{prop:cpband}, not as a separate empirical test of Theorem~\ref{thm:drift}. There we use realized block data to proxy the theorem quantities. Here we use Proposition~\ref{prop:cpband} directly to see how reserve depth and trading frictions move those quantities in a canonical AMM family. We simulate a constant-product pool with fee multiplier $\eta=0.997$, reference-price shocks, and blockwise optimal corrective trades computed from \eqref{eq:qplusstar}--\eqref{eq:Hminuscp}. Reserve depth enters through the local proxy $\ell_n\propto \sqrt{x_n y_n}$, while a fixed execution cost moves the economic no-trade radius.

Because that radius itself changes across specifications, the summary statistic in Figure~\ref{fig:sec4-cpmechanism} is the mean absolute pre-execution log-gap rather than excess outside a fixed band. Under a common disturbance path and a common fee, the mean absolute pre-execution gap falls from $1.15\times 10^{-2}$ to $8.67\times 10^{-3}$ and $6.42\times 10^{-3}$ as the reserve-depth scale rises from $0.5$ to $1.0$ and $2.0$. The right panel extends the same mechanism across a joint grid of reserve depth and fixed execution cost. The lowest gaps occur in the deep-liquidity, low-cost corner, while the highest gaps occur in the shallow-liquidity, high-cost corner. This is exactly the comparative static suggested by Proposition~\ref{prop:cpband}: depth makes corrective trading easier, while higher fixed cost widens the no-trade region.

\begin{figure}[t]
\centering
\includegraphics[width=\columnwidth]{\detokenize{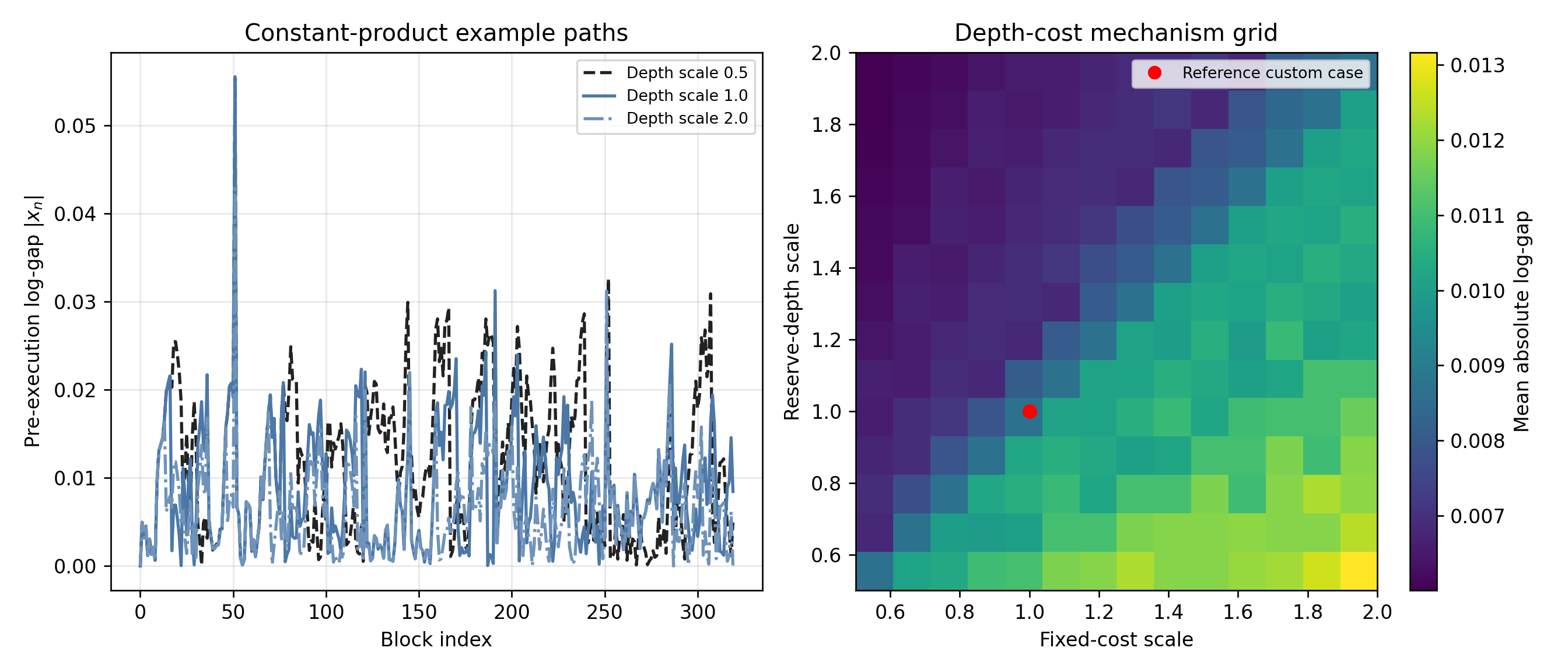}}
\caption{Custom mechanism simulation based on the constant-product example. Left: pre-execution log-gap paths under a common shock sequence for reserve-depth scales $0.5$, $1.0$, and $2.0$. Right: mean absolute pre-execution log-gap on a joint grid of reserve depth and fixed execution cost.}
\label{fig:sec4-cpmechanism}
\end{figure}

This final step matters for interpretation. The reduced simulation is not a separate story from the theorem, and neither is the constant-product mechanism exercise. The first uses data-guided parameter values to organize block-scale correction scenarios. The second moves depth and friction directly inside a canonical AMM family. Together they show the same comparative statics as the theorem without claiming a sharp empirical validation of the theorem itself.

\section{Conclusion}
\label{sec:conclusion}

This paper studies AMM price tracking at block scale with
arbitrage as the corrective input and blockchain execution
inside the loop. The execution layer is reduced to the realized
correction delivered in each block together with a dead-zone
cap. Under a large-error correction condition, this is enough
to prove geometric ergodicity of the tracking error and to
obtain explicit one-step bounds that connect tracking quality
to liquidity and execution quality. The empirical and simulation
sections use proxies for the theorem quantities built from
realized block data and use them to organize reduced and
mechanism-focused simulations. The contribution is therefore
to turn the usual economic intuition that arbitrage keeps AMM
prices aligned into a quantitative stability statement together
with a tractable calibration interface.

\bibliographystyle{IEEEtran}
\bibliography{cdc_first_three_sections}

\end{document}